\documentclass[10pt, conference, letterpaper]{IEEEtran}%
\usepackage{amsmath}
\usepackage{bbm}
\usepackage{graphicx}
\usepackage{epsfig}
\usepackage{array}
\usepackage{psfrag}
\usepackage{amssymb}
\usepackage{mathdots}
\usepackage{color}
\usepackage{pstricks,pst-node,pst-text,pst-3d,pst-plot}
\usepackage{psfrag}
\usepackage{enumerate}
\usepackage{url,cite}
\usepackage{amsfonts}%
\usepackage{amsthm}
\usepackage{dsfont}%
\usepackage{verbatim}%
\usepackage{setspace}
\usepackage{float}
\usepackage{url}
\usepackage{fancyhdr}
\usepackage{bm}
\usepackage{algorithm}
\usepackage{algorithmic}
\newtheorem{thm}{Theorem}
\newtheorem{lma}{Lemma}
\newtheorem{Def}{Definition}
\DeclareMathOperator{\E}{\mathbb{E}}

\newcommand{\lb}{\left (}
\newcommand{\rb}{\right )}

\newcommand{\script}[1]{{\mathcal {#1}}}

\newcommand{\Pmax}{P_{\rm max}}

\newcommand{\Iinst}{I_{\rm inst}}

\newcommand{\fgammai}{f_{\gamma_i}}
\newcommand{\fgi}{f_{g_i}}
\newcommand{\EE}[1]{\E \left[ #1 \right]}
\newcommand{\EEC}[1]{\E_{\bfY(k)} \left[ #1 \right]}
\newcommand{\bgi}{\bar{g}_i}
\newcommand{\bW}{\bar{W}}

\newcommand{\bfP}{{\bf P}}

\newcommand{\bfpi}{{\bm{\pi}}}

\newcommand{\bgamma}{\bar{\gamma}}

\newcommand{\Yivq}{\{Y_i(k)\}_{k=0}^\infty}
\newcommand{\seq}[1]{\{#1\}_{k=1}^\infty}
\newcommand{\bfY}{{\bf Y}}

\newcommand{\Bw}{B_{\rm w}}
\newcommand{\pardef}[1]{\triangleq [#1_1^{(t)},\cdots,#1_N^{(t)}]^T}
\newcommand{\parFdef}[1]{\triangleq [#1_1(k),\cdots,#1_N(k)]^T}
\newcommand{\DOIC}{\emph{DOIC}}

\newcommand{\Ri}{R_i}
\newcommand{\FDurK}{\vert \script{F}_k\vert}

\newcommand{\Ts}{T_{\rm s}}

\begin{document}
\title{Scheduling in Instantaneous-Interference-Limited CR Networks with Delay Guarantees}
\author{Ahmed Ewaisha, Cihan Tepedelenlio\u{g}lu\\
\small{School of Electrical, Computer, and Energy Engineering, Arizona State University, USA.}\\
\small{Email:\{ewaisha, cihan\}@asu.edu}\\
\small{August 2015}
}
\maketitle
\begin{abstract}
We study an uplink multi secondary user (SU) cognitive radio system having average delay constraints as well as an instantaneous interference constraint to the primary user (PU). If the interference channels from the SUs to the PU have independent but not identically distributed fading coefficients, then the SUs will experience heterogeneous delay performances. This is because SUs causing low interference to the PU will be scheduled more frequently, and/or allocated more transmission power than those causing high interference. We propose a dynamic scheduling-and-power-control algorithm that can provide the required average delay guarantees to all SUs as well as protecting the PU from interference. Using the Lyapunov technique, we show that our algorithm is asymptotically delay optimal while satisfying the delay and interference constraints. We support our findings by extensive system simulations and show the robustness of the proposed algorithm against channel estimation errors.

\end{abstract}

\section{Introduction}

The problem of scarcity in the spectrum band has led to a wide interest in cognitive radio (CR) networks. CRs refer to devices that coexist with the licensed spectrum owners called the primary users (PUs). CRs are capable of dynamically adjusting their transmission parameters according to the environment to avoid harmful interference to the PUs. Hence, CR users are required to adjust their transmission power levels, and -thus- their rates, according to the interference level the PUs can tolerate. This adjustment could lead to severe degradation in the quality of service (QoS) provided for the CR users, if not designed carefully.

Moreover, CR users, also referred to as the secondary users (SUs), located physically closer to the PUs might suffer larger degradation in their QoS compared to those that are far because closer SUs transmit with smaller amounts of power. This problem does not appear in conventional non CR cellular systems since frequency channels tend to be orthogonal in non CR systems. In other words, in non CR systems, all users are allocated the channels via some scheduler that guarantees those users do not interfere with each other. While in CR systems, the SUs should be scheduled and have their power controlled in such a way that prevents the harmful interference to the PUs since they share the spectrum.

The problem of scheduling and/or power control for CR systems has been widely studied in the literature (please see \cite{Letaief_PU_Known_Location,NEP_Distributed,Ewaisha_TVT2015,Iter_Bit_Allocation_OFDM,6464638}, and references therein). 
The algorithms proposed in these works aim at optimizing the throughput for the SUs and, at the same time, protecting the PUs from interference. However, providing guarantees on the queuing delay in CR systems was not the goal of these works. In real-time applications, such as audio/video conference calls, packets are expected to arrive at the destination before a prespecified deadline. Thus, the average packet delay needs to be as small as possible to prevent jitter and to guarantee acceptable QoS for these applications \cite{shakkottai2002scheduling,kang2013performance}.

Queuing delay has gained strong attention recently in the literature and scheduling algorithms have been proposed to guarantee small delay \cite{li2011delay,Two_Q_Light_Hvy,neely2003power}. In \cite{li2011delay}, the authors study the joint scheduling-and-power-allocation problem in the presence of an average power constraint. Although in \cite{li2011delay} the proposed algorithm offers an acceptable delay performance, all users are assumed to transmit with the same power. A power allocation and routing algorithm is proposed in \cite{neely2003power} to maximize the capacity region under an instantaneous power constraint. While the authors show an upper bound on the average delay, this delay performance is not guaranteed to be optimal.

Although queuing theory, originally developed to model packets at a server, can be applied to wireless channels, the challenges are different. One of the main challenges is the fading nature of the wireless channel that changes from a slot to another. Fading requires adapting the user's power and/or rate according to the channel's fading coefficient. The idea of power and/or rate adaptation based on the channel condition does not have an analogy in server problems and, thus, is absent in the aforementioned references. Instead, existing works treat wireless channels as on-off fading channels and do not consider multiple fading levels. Among the relevant references that consider a more general fading channel model are \cite{neely2003power}, which was discussed above, \cite{Fading_No_Scheduling,E_Hossain_CR_Delay_Analysis} where the optimization over the scheduling algorithm was out of the scope of their work, and \cite{Min_Pow_4_Delay_NonCR} that neglects the interference constraint since it considers a non CR system.

In contrast with \cite{Letaief_PU_Known_Location,NEP_Distributed,Ewaisha_TVT2015,Iter_Bit_Allocation_OFDM,6464638} that do not optimize the queuing delay, the problem of minimizing the sum of SUs' average delays is considered in this paper. The proposed algorithm guarantees a bound on the instantaneous interference to the PUs, a guarantee that is absent in \cite{li2011delay,neely2003power}. Based on Lyapunov optimization techniques \cite{li2011delay}, an algorithm that dynamically schedules the SUs as well as optimally controlling their transmission power is presented. The contributions in this paper are: i) Proposing a joint power-control and scheduling algorithm that is optimal with respect to the average delay of the SUs in an interference-limited system; ii) Showing that the proposed algorithm can provide differentiated service to the different SUs based on their heterogeneous QoS requirements. Moreover, the complexity of the algorithm is shown to be polynomial in time since it is equivalent to that of sorting a vector of $N$ numbers, where $N$ is the number of SUs in the system.



The rest of the paper is organized as follows. The network model and the underlying assumptions are presented in Section \ref{Model}. In Section \ref{Prob_Statement} we formulate the problem mathematically. The proposed algorithm, its optimality and complexity are presented in Section \ref{Proposed_Algorithm},
followed by the extensive simulation results in Section \ref{Results}. Finally the paper is concluded in Section \ref{Conclusion}.

\section{Network Model}
\label{Model}
We assume a CR system consisting of a single secondary base station (BS) serving $N$ secondary users (SUs) indexed by the set $\script{N}\triangleq \{1,...N\}$. We are considering the uplink phase where each SU has its own queue buffer for packets that need to be sent to the BS. The SUs are assumed to share a single frequency channel with a single PU that has licensed access to this channel. The CR system is assumed to operate in an underlay fashion where SUs are allowed to transmit as long as the power received by the PU from their transmission does not exceed a prespecified threshold $\Iinst$ at any given slot. Moreover, in order for the BS to be able to decode the received signal, no more than one SU at a time slot is to be assigned the channel for transmission.

\subsection{Channel and Interference Model}
We assume a time slotted structure where each slot is of duration $\Ts$ equal to the coherence time of the channel. The channel between SU $i$ and the BS is assumed to be block fading with instantaneous power gain $\gamma_i^{(t)}$, at time slot $t$, following the probability density function (PDF) $\fgammai(\gamma)$ with mean $\bgamma_i$. The channel gain is assumed to be independent across SUs but not necessary identically distributed
. 
SUs can use a rate adaptation scheme selected based on the channel gain $\gamma_i^{(t)}$. The transmission rate of SU $i$ at time slot $t$ is\footnote{All logarithms in this paper have base $e$, i.e. are natural logarithms.}
\begin{equation}
\Ri\lb P_i^{(t)}\rb=\Bw \log \lb 1+P_i^{(t)}\gamma_i^{(t)} \rb \hspace{0.1in} \rm{packets/sec},
\label{Tx_Rate}
\end{equation}
where $\Bw$ is the bandwidth of the channel, while $P_i^{(t)}$ is the power by which SU $i$ transmits its packet at slot $t$. For a fixed power allocation $P_i^{(t)}=P_i$, we define the service rate of SU $i$ as $\mu_i(P_i)$ where $1/\mu_i(P_i)\triangleq \EE{1/\Ri \lb P_i\rb}$ is the average time required to serve one packet for SU $i$ transmitting with power $P_i$. We assume that $\gamma_i^{(t)}$ has a distribution that results in finite values for the first four moments of the service time $1/\Ri \lb P_i\rb$. That is, $\EE{\lb 1/\Ri \lb P_i\rb\rb^n}<\infty$, $n=1,2,3,4$, $\forall i \in \script{N}$.

We assume that the single PU is in the vicinity of the CR system and is transmitting at all times. This PU suffers interference from the SUs through the channel between each SU and this PU. The interference channel between SU $i$ and the PU, at slot $t$, has a power gain $g_i^{(t)}$ following a PDF $\fgi(g_i)$ with mean $\bgi$. The power gains are assumed to be independent among SUs but not identically distributed. We assume that SU $i$ knows the value of $\gamma_i^{(t)}$ as well as $g_i^{(t)}$, at the beginning of slot $t$ through some channel estimation phase \cite{haykin2005cognitive}. The channel estimation phase to acquire $g_i^{(t)}$ is done by overhearing the pilots transmitted by the primary transmitter to its intended receiver. The channel estimation phase is out of the scope of this work, however the effect of channel estimation errors will be discussed in Section \ref{Results}.

\subsection{Queuing Model}
We assume that each SU $i$ has an M/G/1 queuing system. The number of packets arriving per unit of time to the SU's buffer follows a Poisson process with a fixed parameter $\lambda_i$ packets per second. Each packet has a fixed length of $L$ bits that is constant for all users. However, each packet takes a random amount of time to be transmitted to the BS that depends on the rate of transmission $\Ri( P_i^{(t)} )$. Thus, the service time will be assumed to follow a general distribution throughout the paper that depends on the distribution of $\gamma_i^{(t)}$.

We assume that the packets arriving to the SUs are buffered in infinite-sized buffers and are served according to the first-come-first-serve discipline. Thus when SU $i$ is scheduled for transmission at slot $t$, it transmits the first $M_i^{(t)}$ packets of its queue, where
\begin{align}
\nonumber M_i^{(t)}&=\min \lb \left\lfloor \Ts \Ri\lb P_i^{(t)} \rb \right\rfloor,Q_i^{(t)} \rb \hspace{0.1in} \rm{packets}\\
&\approx \min\lb \Ts \Ri\lb P_i^{(t)} \rb,Q_i^{(t)}\rb,
\label{Num_Pkts}
\end{align}
where $Q_i^{(t)}$ is the number of packets buffered at SU $i$ at the beginning of slot $t$ and is given by
\begin{equation}
Q_i^{(t+1)}\triangleq \lb Q_i^{(t)} + \vert \script{A}_i^{(t)}\vert - M_i^{(t)} \rb^+,
\label{Queue}
\end{equation}
where $\script{A}_i^{(t)}$ is the set of all packets arrived to SU $i$ at slot $t$, while $\vert \script{A}_i^{(t)}\vert$ is the number of elements in this set.\footnote{Our model is also valid if the arrivals follow a discrete time process where $\script{A}_i^{(t)}$ will, then, represent the packets arrived at the beginning of slot $t$.} The approximation in \eqref{Num_Pkts} is valid when the parameter $\Bw \Ts\gg 1$. This is because $\lfloor x \rfloor \approx x$ when $x$ is a large number.

We define the delay $W_i^{(j)}$ of a packet $j$ as the total amount of time, in time slots, packet $j$ spent in SU $i$'s buffer from the instant it joined the queue until the slot it is transmitted in\footnote{We do not include the transmission slot when calculating $W_i^{(j)}$, but we include the residual time which is non-zero since packets are allowed to arrive in the middle of a time slot.}. The time-average delay experienced by SU $i$'s packets is given by \cite{li2011delay}
\begin{equation}
\bW_i \triangleq \lim_{T \rightarrow \infty} \frac{\EE{\sum_{t=1}^T{\sum_{j\in\script{A}_i^{(t)}}{W_i^{(j)}}}}}{\EE{\sum_{t=1}^T{\vert \script{A}_i^{(t)}\vert}}}
\label{Delay}
\end{equation}
which is the expected total amount of time spent by all packets arriving in a time interval, of an large duration, normalized by the expected number of packets arrived in this interval.

\subsection{Frame-Based Policy}
\label{Frame_Based_Policy}
Without losing optimality (Lemma 1 of \cite{li2011delay}), one can use a ``frame-based'' random priority list as the scheduling policy. This policy schedules the SUs based on a strict non-preemptive priority list and updates this list, at the beginning of each frame, randomly following some optimum, genie-aided, distribution defined over the $N!$ distinct priority lists. Hence, we divide time into frames where frame $k$ consists of $\vert \script{F}(k)\vert$ consecutive time-slots, where $\script{F}(k)$ is the set containing the indices of the time slots belonging to frame $k$ (see Fig. \ref{Frame_Structure}). 
Since frames do not overlap, if $t\in\script{F}(k_1)$ then $t \notin \script{F}(k_2)$ as long as $k_1\neq k_2$. One of our goals in this paper is to choose, at the beginning of each frame $k$, the ``best'' priority list so that the system has the same long-term delay performance as this, genie-aided, random priority list. We note that the average delay $\bW_i$ in \eqref{Delay} depends on the chosen priority lists of all frames $k=1,\cdots \infty$. This dependency is implicitly mentioned in the expected value operator $\EE{\cdot}$. We now define how a frame begins and ends. Each frame consists of exactly one \emph{idle period} followed by exactly one \emph{busy period}, both are defined next.

\begin{figure}%
\includegraphics[width=1\columnwidth]{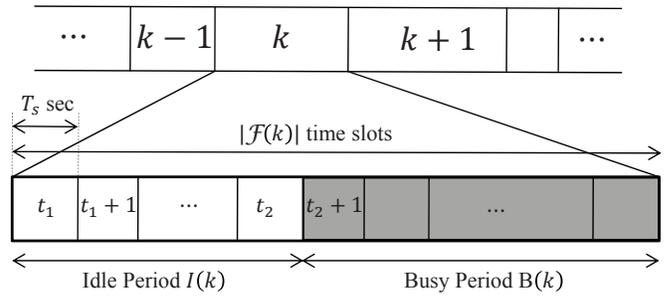}%
\caption{Time is divided into frames. Frame $k$ has $\FDurK$ slots each is of duration $\Ts$ seconds. Different frames can have different number of time slots.}%
\label{Frame_Structure}%
\end{figure}

\begin{Def}
\label{Idle_Def}
Idle period is the period formed by the consecutive time slots where all SUs have empty buffers. An idle period starts with the time slot $t_1$ following the transmission of the last packet in the system, and ends with a time slot $t_2$ when one of the SUs' buffer receives one or more new packets to be transmitted (see Fig. \ref{Frame_Structure}). In other words, $t_1$ satisfies $\sum_{i\in\script{N}}Q_i^{(t_1)}=0$ and $\sum_{i\in\script{N}}Q_i^{(t_1-1)}\neq 0$, while $t_2$ satisfies $\sum_{i\in\script{N}}Q_i^{(t_2)}=0$ and $\sum_{i\in\script{N}}Q_i^{(t_2+1)}\neq 0$.
\end{Def}

\begin{Def}
\label{Busy_Def}
Busy period is the period between two consecutive idle periods.
\end{Def}

The duration of the idle period $I(k)$ and busy period $B(k)$ of frame $k$ are random variables, thus the cardinality $\vert \script{F}(k) \vert=I(k)+B(k)$ is random as well. 

\subsection{Transmission Process}
At the beginning of frame $k$, the BS finds the priority list $\bfpi(k) \parFdef{\pi}$ where $\pi_j(k)$ is the index of the SU who will be assigned priority $j$ during frame $k$. And at the beginning of each time slot $t\in \script{F}(k)$, the BS schedules the user with the highest priority in the list $\bfpi(k)$, among all SUs having non-empty buffers. Then, it broadcasts its index, say index $i^*$, on a common control channel as well as broadcasting the power $P_{i^*}^{(t)}$ by which this SU will be transmitting during slot $t$. SU ${i^*}$, in turn, begins transmission of the first $M_{i^*}^{(t)}$ packets with a constant power $P_{i^*}^{(t)}$. We assume the BS receives these packets error-free by the end of slot $t$ then a new time slot $t+1$ starts.

\section{Problem Statement}
\label{Prob_Statement}
Each SU $i$ has an average delay constraint $\bW_i \leq d_i$ that needs to be satisfied. Moreover, the PU can tolerate an interference level of $\Iinst$ at any given slot. Hence, the main objective is to solve the following problem
\begin{equation}
\begin{array}{ll}
\underset{\{\bfP^{(t)}\},\{\bfpi(k)\}}{\rm{minimize}}& \sum_{i=1}^N h_i\lb\bW_i\rb\\
\label{Problem}
\rm{subject \; to} & \sum_{i=1}^N {P_i^{(t)}g_i^{(t)}} \leq \Iinst \hspace{0.25in} , \hspace{0.25in} \forall t\geq 1\\
& \bW_i \leq d_i\\
& P_i^{(t)} \leq \Pmax \hspace{0.1in} , \hspace{0.1in} \forall i \in \script{N} \rm{\; and \;}\forall t \geq 1,\\
& \sum_{i=1}^N{  \mathds{1} \lb P_i^{(t)}\rb} \leq 1 \hspace{0.25in}, \hspace{0.25in} \forall t \geq 1,
\end{array}
\end{equation}
where $\bfP^{(t)} \pardef{P}$, $h_i\lb \cdot \rb$ are some convex increasing functions called the ``cost functions'', while $\mathds{1}(x)\triangleq1$ if $x\neq 0$ and $0$ otherwise. The last constraint indicates that no more than a single SU is to be transmitting at slot $t$.

As pointed out in Section \ref{Frame_Based_Policy}, the average delay $\bW_i$ in equation \eqref{Delay} is a function of the priority lists $\seq{\bfpi(k)}$. Hence, the objective function of \eqref{Problem} is a complicated function in the priority lists $\seq{\bfpi(k)}$. This makes problem \eqref{Problem} a joint power allocation and scheduling problem that is difficult solve using conventional convex, or non-convex, optimization problem algorithms. Thus, we next propose a low complexity update algorithm and show its optimality.

\section{Proposed Power Allocation and Scheduling Algorithm}
\label{Proposed_Algorithm}
Since the complexity of the solution is intractable, we solve the problem by proposing an online algorithm that dynamically updates the power allocation vector $\bfP^{(t)}$ and the priority vector $\bfpi(k)$. We show that this algorithm has an asymptotically optimal performance. That is, we can achieve a delay arbitrarily close to the optimal value depending on some control parameter $V$.

\subsection{Satisfying Delay constraints}
As will be discussed in Section \ref{Algorithm_Subsection}, the proposed algorithm is executed at the beginning of each frame. In order to guarantee a feasible solution satisfying the delay constraints in problem \eqref{Problem}, we set up a ``virtual queue'' associated with each delay constraint in problem \eqref{Problem}. The virtual queue for SU $i$ at frame $k$ is given by
\begin{equation}
Y_i(k+1)=\lb Y_i(k)+\sum_{j\in \script{A}_i(k)}{\lb W_i(j)-r_i(k)\rb} \rb^+
\label{Delay_Q}
\end{equation}
where $r_i(k)\in[0,d_i]$ is an auxiliary variable, that is to be optimized over, while $\script{A}_i(k)\triangleq\cup_{t\in\script{F}(k)}\script{A}_i^{(t)}$ is the set of all packets arrived at SU $i$'s buffer during frame $k$. We define $\bfY(k) \parFdef{Y}$ for notational convenience. Equation \eqref{Delay_Q} is calculated at the end of frame $k$ and represents the amount of delay exceeding the delay bound $d_i$ up to the end of frame $k$. We first mention the following definition, then state a lemma that gives a sufficient condition for the delay of SU $i$ to satisfy $\bW_i \leq d_i$.
\begin{Def}
\label{Mean_Rate_Def}
We say that the random sequence $\{Y_i(k)\}_{k=0}^\infty$ is mean rate stable if and only if $\lim_{K\rightarrow\infty}\EE{Y_i(K)}/K=0$.
\end{Def}

\begin{lma}
\label{Mean_Rate_Lemma}
If $\{Y_i(k)\}_{k=0}^\infty$ is mean rate stable, then the time-average delay of SU $i$ satisfies $\bW_i \leq d_i$.
\end{lma}
\begin{proof}
Lemma 3 of \cite{li2011delay} can be modified to prove that
\begin{equation}
\bW_i \leq \lim_{K \rightarrow \infty} \frac{\EE{\sum_{k=1}^K{\vert \script{A}_i(k) \vert {r_i(k)}}}}{\EE{\sum_{k=1}^T{\vert\script{A}_i(k)\vert}}}.
\label{Wait_r_i}
\end{equation}
The proof follows by replacing $r_i(k)$ by its bound $d_i$ in \eqref{Wait_r_i}.
\end{proof}
Lemma \ref{Mean_Rate_Lemma} says that if the power allocation and scheduling algorithm results in a mean rate stable $\Yivq$, then the average delay constraint of problem \eqref{Problem} is satisfied.

\subsection{Algorithm}
\label{Algorithm_Subsection}
We now propose the \emph{Delay Optimal with Instantaneous Interference Constraint} (\emph{DOIC}) Algorithm executed at the beginning of each frame $k$ for finding $\bfP^{(t)}$ as well as the optimum list $\bfpi(k)$, given some prespecified control parameter $V$.

\noindent{\bf DOIC Algorithm:}
\begin{enumerate}
	\item At the beginning of frame $k$, the BS sorts the SUs according to the descending order of $Y_i(k)\mu_i(k)$. The sorted list is denoted $\bfpi(k)$.
	\item At the beginning of each slot $t\in\script{F}(k)$ the BS schedules SU $i^*$ that has the highest priority in the list $\bfpi(k)$ among those having non-empty buffers.
	\item SU $i^*$, in turn, transmits $M_{i^*}^{(t)}$ packets as dictated by equation \eqref{Num_Pkts} where $P_i^{(t)}=0$ for all $i\neq i^*$ while
	\begin{equation}
	P_{i^*}^{(t)}=\min(\Iinst/g_{i^*}^{(t)},\Pmax).
	\label{Pow_Allocation}
	\end{equation}
	\item At the end of frame $k$, for all $i\in \script{N}$ the BS updates:
	\begin{enumerate}
		\item $Y_i(k+1)$ via equation \eqref{Delay_Q}. And,
		\item $r_i(k+1)\triangleq \arg\underset{r_i(k)}{\min}{Vh_i\lb r_i(k) \rb - Y_i(k)\lambda_i r_i(k)}$.
	\end{enumerate}
	\item Set $k\leftarrow k+1$ and go to Step 1.
\end{enumerate}
 We next discuss the optimality of {\DOIC} in Theorem \ref{Optimality}.

\begin{thm}
\label{Optimality}
When the BS executes \DOIC, the time average of the SUs' delays satisfy the following inequality
\begin{equation}
\sum_{i=1}^N{h_i\lb\bW_i\rb} \leq \frac{C \sum_{i=1}^N{\lambda_i}}{V} + \sum_{i=1}^N{h_i \lb \bW_i^*\rb}
\label{Optimality_Equation}
\end{equation}
where $\bW_i^*$ is the optimum value of the delay when solving problem \eqref{Problem}, while $C$ is some constant that depends on the statistics of the arrival process as well as the channel statistics. Moreover, the virtual queues $\Yivq$ are mean rate stable $\forall i \in \script{N}$.
\end{thm}

\begin{proof}
See Appendix \ref{Optimality_Proof}.
\end{proof}

Theorem \ref{Optimality} says that the objective function of problem \eqref{Problem} is upper bounded by the sum of the optimum values $h_i \lb \bW_i^*\rb$ plus some constant that vanishes as $V\rightarrow\infty$. The drawback of setting $V$ very large is that the algorithm converges slower. That is, the virtual queues become mean rate stable after a larger number of frames. Having a vanishing gap means that {\DOIC} is asymptotically optimal. Moreover, based on the mean rate stability of the queues $\Yivq$, the delay constraints of problem \eqref{Problem} are satisfied. We note that the complexity of the {\DOIC} algorithm is polynomial in time since it only requires sorting the quantity $Y_i(k) \mu_i(k)$, $i \in \script{N}$, in a descending order.

\section{Simulation Results}
\label{Results}
We simulated the system of $N=2$ SUs with $h_i(x)=x^2/2$ $\forall i=1,2$, which is a cost function that guarantees proportional fairness\footnote{See Chapter 2.2 of \cite{srikant2013communication}.} among SUs (refer to Table \ref{Parameters} for a complete list of parameter values). Since the interference channel gain for SU $1$ is higher than that for SU $2$, we expect SU $1$ to have higher average delay. However, the {\DOIC} algorithm can guarantee a bound on this delay using the constraint $\bW_1\leq d_1$, so that the QoS requirement of SU $1$ is satisfied. In our simulations we set $d_1=1.25\Ts$.

Assuming perfect knowledge of the direct and interference channel state information (CSI), namely $g_i^{(t)}$ and $\gamma_i^{(t)}$, Fig. \ref{PerUser_Delay_Inst} plots the average per-SU delay $\bW_i$, from equation \eqref{Delay}, for two scenarios; the first being the constrained optimization problem where $d_1=1.25\Ts$ while setting $d_2$ to any arbitrary high value (we set $d_2=3\Ts$), while the second is the unconstrained optimization problem\footnote{The reason we call it the unconstrained problem is because the average delay of both SUs is strictly below $3\Ts$, thus both delay constraints are inactive.} where both $d_1$ and $d_2$ are set arbitrary high (we set $d_1=d_2=3\Ts$). The X-axis is the average number of packets arriving per time slot $\lambda$, where $\lambda \triangleq \lambda_1=\lambda_2$. 
From Fig. \ref{PerUser_Delay_Inst} we can see a gap, in the unconstrained problem, between the average delay of SU $1$ and that of SU $2$. Hence, SU $1$ suffers arbitrary high delay. While for the constrained problem, the {\DOIC} has forced $\bW_1$ to be smaller than $1.25\Ts$ for all $\lambda$ values. This comes at the cost of SU $2$'s delay. We conclude that the delay constraints in problem \eqref{Problem} can force the delay vector of the SUs to take any value as long as it is feasible.

Fig. \ref{Avg_Delay_Inst} plots the sum of cost functions versus $\lambda$ for the perfect CSI case for both the constrained and unconstrained problem. Comparing the two curves, we find that the constrained problem has worse sum of cost functions for $\lambda\geq 0.6$. This is because the constrained problem has a smaller feasible region than that of the unconstrained one. We note that the PU's interference constrained is satisfied with probability 1 based on equation \eqref{Pow_Allocation}.

For the imperfect CSI case, we assumed that each SU has an error of $10\%$ in estimating each of $g_i^{(t)}$ and $\gamma_i^{(t)}$. And to guarantee protection to the PU from interference, we replaced equation \eqref{Pow_Allocation} by $P_{i^*}^{(t)}=\min(\Iinst/g_{i^*}^{(t)}/1.1,\Pmax)$. From Fig. \ref{Avg_Delay_Inst} we can see that the performance difference between the perfect and the imperfect CSI constrained problem ranges between $5.6\%$ at $\lambda=0.1$, and $20\%$ at $\lambda=1$. This is considered a good performance for this high CSI estimation error value. We may note that this performance difference represents an upper bound on the actual difference since the $10\%$ is an upper bound on the actual estimation error.

\begin{figure}%
\centering
\includegraphics[width=\columnwidth]{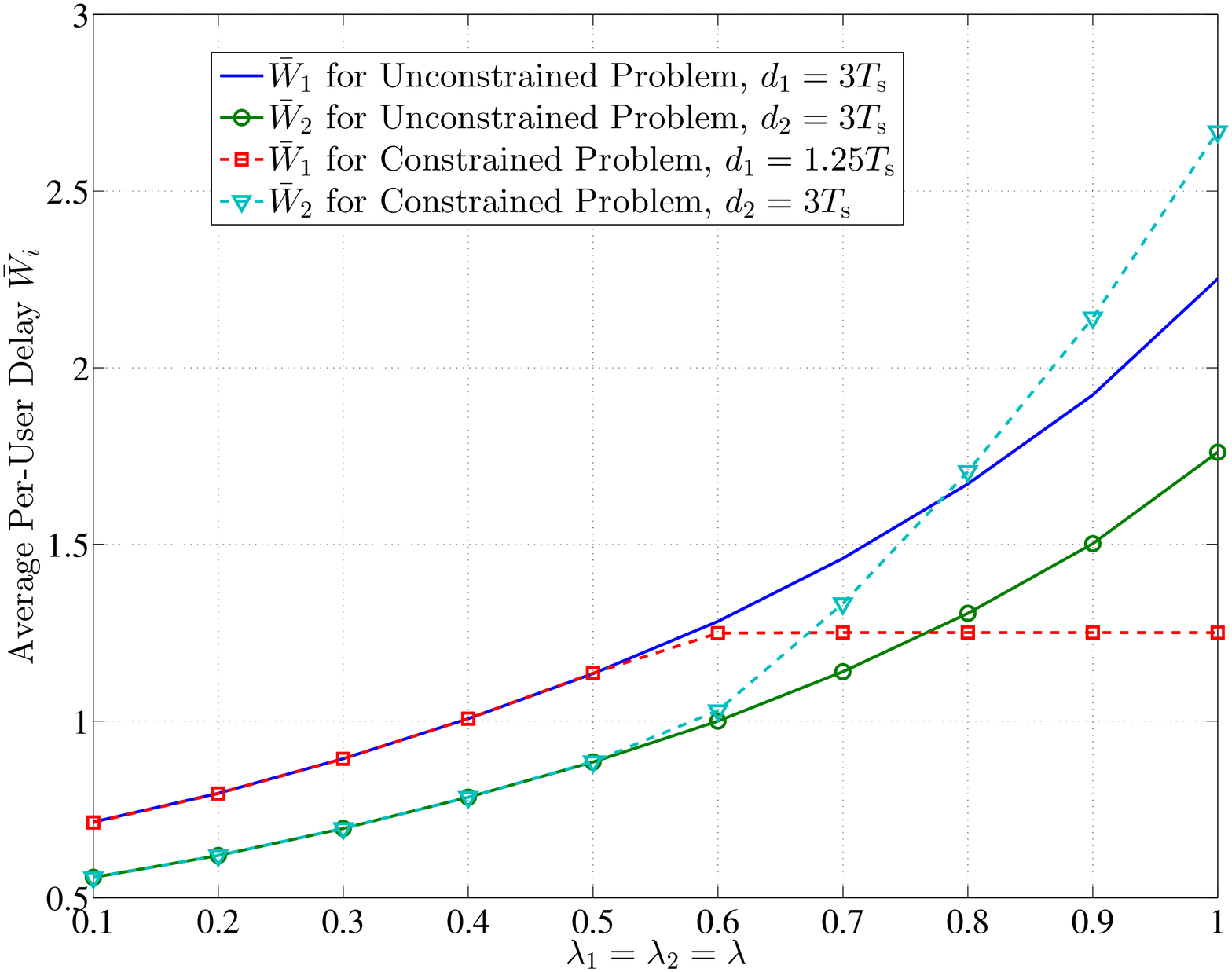}%
\caption{Average per-user delay for both the constrained and unconstrained optimization problems}%
\label{PerUser_Delay_Inst}%
\end{figure}

\begin{figure}%
\centering
\includegraphics[width=\columnwidth]{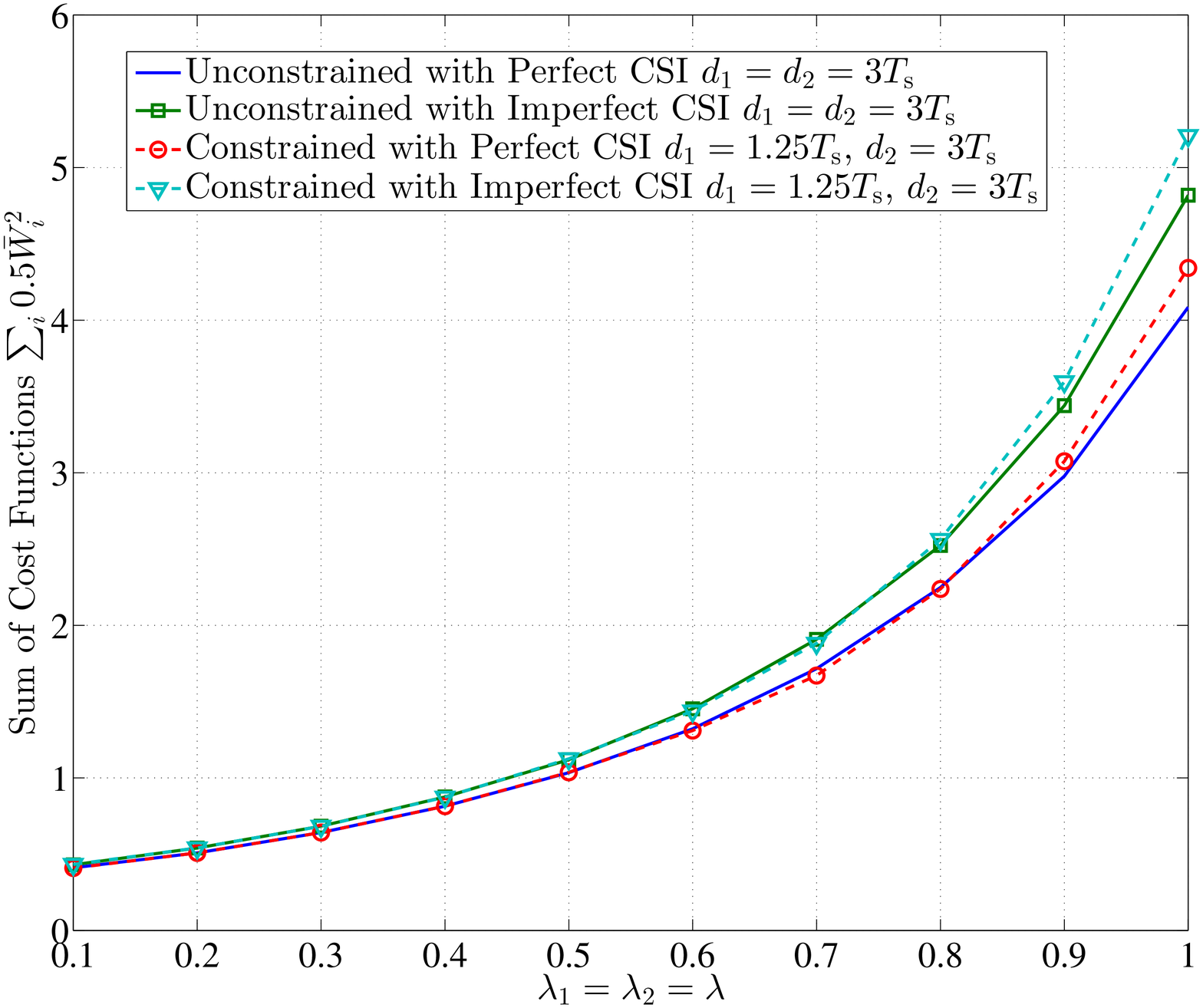}%
\caption{Sum of cost functions for the perfect as well as the imperfect channel sensing for both the constrained and unconstrained optimization problems.}%
\label{Avg_Delay_Inst}%
\end{figure}

\begin{table}
	\centering
		\caption{Simulation Parameter Values}
		\label{Parameters}
		\begin{tabular}{|c|c|}
			\cline{1-2}
			Parameter & Value \\ \cline{1-2}
			$\lambda_1=\lambda_2=\lambda$ & $\lambda \in \{0.1,0.2,\cdots,1\}$ packets/slot\\
			$\fgammai(\gamma)$ & $\exp{\lb-\gamma/\bgamma_i\rb}/\bgamma_i$\\
			$\fgi(g)$ & $\exp{\lb-g/\bar{g}_i\rb}/\bar{g}_i$\\
			$(\bgamma_1,\bgamma_2)$ & $(1,1)$\\
			$(\bar{g}_1,\bar{g}_2)$ & $(4,2)$\\
			$\Iinst$ & 5 \\
			$\Pmax$ & 10 \\
			$\epsilon$ & $10^{-2}$ \\
			$V$ & $10^3$ \\ \cline{1-2}
			\end{tabular}
\end{table}

\section{Conclusion}
\label{Conclusion}
We have studied the joint scheduling and power allocation problem of an uplink multi SU CR system. We formulated the problem as a delay minimization problem in the presence of instantaneous interference constraints to the PU, as well as an average delay constraint for each SU that needs to be met. Most of the existing literature that studies this problem either assumes on-off fading channels or does not provide a delay-optimal algorithms which is essential for real-time applications.

We proposed a dynamic algorithm that schedules the SUs based on a dynamically-updated priority list. The proposed algorithm updates the priority list on a frame basis while controlling the power on a per-slot basis. 
We showed, through the Lyapunov optimization, that the proposed algorithm is asymptotically delay optimal. That is, it minimizes the sum of any convex increasing function of the average delays of the SUs as well as satisfying the instantaneous interference and average delay constraints. Extensive simulation results showed the robustness of our algorithm against CSI estimation errors.


\appendices
\section{Proof of Theorem \ref{Optimality}}
\label{Optimality_Proof}
We define the Lyapunov function and Lyapunov drift to be
\begin{align}
&L \lb \bfY(k)\rb \triangleq \frac{1}{2}\sum_{i=1}^N Y_i^2(k), \hspace{0.2in}{\rm and}
\label{Lyap_Func}\\
\Delta \lb \bfY(k)\rb &\triangleq \EEC{L \lb \bfY(k+1) \rb - L \lb \bfY(k) \rb},
\label{Drift_Def}
\end{align}
respectively, where $\EEC{x}$ denotes the conditional expectation of $x$ given $\bfY(k)$, namely $\EEC{x} \triangleq \EE{x \vert \bfY(k)}$. Squaring equation \eqref{Delay_Q} and taking the conditional expectation, we get
\begin{align}
\nonumber &\EEC{Y_i^2(k+1)-Y_i^2(k)} =Y_i(k) \lb \EEC{W_i(k)\FDurK}\lambda_i\rb\\
& -\EEC{\FDurK}+ \EEC{\lb \sum_{j\in \script{F}_k} \lb W_i^{(j)} - r_i(k)\rb\rb^2}
\lambda_i r_i(k)
\label{Delay_Q_Sq1}\\
& \hspace{0.3in} \leq Y_i(k) \EEC{\FDurK}\lambda_i \lb \bW_i(k)-r_i(k)\rb + C,
\label{Delay_Q_Sq2}
\end{align}
where the inequality in \eqref{Delay_Q_Sq2} comes from upper-bounding the first term in \eqref{Delay_Q_Sq1} by some finite constant $C<\infty$. The existence of this finite constant can be shown by following steps similar to the proof of Lemma 4 in \cite{li2011delay}. We omit these steps for brevity. Given some fixed control parameter $V>0$, we add the penalty term $V\sum_i \EEC{h_i(r_i(k))\FDurK}$ to both sides of \eqref{Drift_Def}. Using the bound in \eqref{Delay_Q_Sq2}, the drift-plus-penalty term becomes bounded by
\begin{align}
\nonumber \Delta &\lb \bfY(k)\rb + V\sum_{i=1}^N \EEC{h_i(r_i(k))\FDurK}
\leq \EEC{\FDurK} \times \\
&\sum_{i=1}^N  \left[ V h_i(r_i(k))  +  Y_i(k) \lambda_i \lb \bW_i(k)- r_i(k)\rb \right] +C.
\label{Drift_Plus_Penalty1}
\end{align}
The proposed {\DOIC} algorithm minimizes the summation in the right-side of \eqref{Drift_Plus_Penalty1}. This gives a lower bound on the right-side of \eqref{Drift_Plus_Penalty1} under any other algorithm. That is, the right-side, evaluated under the \DOIC, is a lower bound on right-side when evaluated under any other policy including the optimal policy that solves \eqref{Problem}. Hence, we now evaluate the right-side of \eqref{Drift_Plus_Penalty1} at the optimal policy that solves \eqref{Problem} with the help of a genie-aided knowledge of $r_i(k)=\bW_i^*$, yielding an upper bound on the drift-plus-penalty evaluated at \DOIC. Namely
\begin{align}
\nonumber \Delta \lb \bfY(k)\rb +& V\sum_{i=1}^N \EEC{h_i(r_i(k))\FDurK} \\
&\leq C + V\sum_{i=1}^N h_i(\bW_i^*)\EEC{\FDurK}
\label{DOIC_Genie}
\end{align}
where the left-side of \eqref{DOIC_Genie} is evaluated at \DOIC. Taking $\EE{\cdot}$, summing over $k=0,\cdots,K-1$, denoting $L\lb \bfY_i(0)\rb\triangleq 0$ for all $i\in\script{N}$, and dividing by $V\sum_{k=1}^{K-1} \EE{\FDurK}$ we get
\begin{align}
\nonumber \sum_{i=1}^N \frac{\EE{Y_i^2(K)}}{K}& \frac{K}{\sum_k \EE{\FDurK}}+ \sum_{i=1}^N \frac{\sum_k\EE{h_i \lb r_i(k) \rb\FDurK}}{\sum_k \EE{\FDurK}}\\
& \leq \frac{CK}{V\sum_k \EE{\FDurK}} + \sum_{i=1}^N h_i\lb \bW_i^*\rb.
\label{Optimal_Eq}
\end{align}
From equation (3) in \cite{li2011delay} we have
\begin{equation}
\EE{\FDurK}=\frac{1}{\lb 1-\sum_i \frac{\lambda_i}{\EE{\mu_i(P_i^{(t)}}}\rb \lb 1-\sum_i \lambda_i \rb},
\label{Frame_Bound}
\end{equation}
where $P_i^{(t)}$ is given by step 3 in the {\DOIC} algorithm. To prove the upper bound in theorem \ref{Optimality}, we remove the first summation in the left-side of \eqref{Optimal_Eq}, and use Lemma 7.6 in \cite{neely2010stochastic} to prove that
\begin{equation}
h_i \lb \bW_i \rb \leq \frac{\sum_k\EE{h_i \lb r_i(k) \rb\FDurK}}{\sum_k \EE{\FDurK}},
\label{Lemma_Seven_Point_Six}
\end{equation}
where $\bW_i$ is given by equation \eqref{Delay}. On the other hand, to prove the mean rate stability of the sequence $\{Y_i(k)\}_{k=0}^\infty$, we use \eqref{Frame_Bound} to denote the right-side of \eqref{Optimal_Eq} by $C_1$ (since it does not depend on $K$), remove the second summation in the left-side of \eqref{Optimal_Eq}, use equation \eqref{Frame_Bound} to obtain
\begin{equation}
\sum_{i=1}^N \frac{\EE{Y_i^2(K)}}{K} \leq C_1,
\label{MRS_bound}
\end{equation}
and use Jensen's inequality to note that
\begin{equation}
\frac{\EE{Y_i(K)}}{K} \leq \sqrt{\frac{\EE{Y_i^2(K)}}{K^2}} \leq \sqrt{\frac{C_1}{K}}.
\label{Jensens}
\end{equation}
Finally, taking the limit when $K\rightarrow \infty$ completes the mean rate stability proof.
%



\bibliographystyle{IEEEbib}
\bibliography{MyLib}

\end{document}